\documentclass[a4paper,UKenglish,cleveref, 
autoref,thm-restate]{lipics-v2021}

\hideLIPIcs  

\title{Self-Stabilizing MIS Computation in the Beeping Model} 

\author{George {Giakkoupis}}{Inria, Univ Rennes, CNRS, IRISA, Rennes, France}{george.giakkoupis@inria.fr}{https://orcid.org/0000-0002-8023-4485}{Supported by Agence Nationale de
la Recherche (ANR), under project ByBloS (ANR-20-CE25-0002).}

\author{Volker {Turau}}{Institute of Telematics, Hamburg University of Technology, Hamburg, Germany}{turau@tuhh.de}{https://orcid.org/0000-0001-9964-8816}{}

\author{Isabella {Ziccardi}}{Bocconi University, BIDSA, Milan, Italy}{isabella.ziccardi@unibocconi.it}{https://orcid.org/0000-0002-1550-3677}{Supported by the European Research Council (ERC), under the European Union’s Horizon 2020 research and innovation program (grant agreement No. 834861).}

\authorrunning{G. Giakkoupis, V. Turau and I. Ziccardi} 

\ccsdesc[300]{Theory of computation~Distributed algorithms}

\keywords{Maximal Independent Set, Self-Stabilization, Beeping Model} 

\nolinenumbers 

\usepackage[T1]{fontenc}    
\usepackage{amsfonts,amsmath,amssymb,dsfont,mathtools}

\usepackage[ruled,vlined]{algorithm2e}
\usepackage{dsfont}
\usepackage[textsize=tiny]{todonotes}

\crefname{equation}{}{}
\crefname{claim}{Claim}{Claims}

\newcommand{\ellmax}{\ell_{\max}}

\SetKwInput{state}{state}
\SetKwInput{init}{initialization}
\SetKwInput{output}{output}
\SetKwFor{InEachRound}{in each round}{do}{end}
\SetKwFor{InEachOddRound}{in each odd round}{do}{end}
\SetKwFor{InEachEvenRound}{in each even round}{do}{end}
\DontPrintSemicolon

\newcommand{\false}[0]{\ensuremath{\mathit{false}}\xspace}
\newcommand{\true}[0]{\ensuremath{\mathit{true}}\xspace}

\newcommand{\mis}{\mathit{I}}
\newcommand{\Promi}{\mathit{PM}}
\newcommand{\Expc}[1]{\mathbf{E}\left[#1\right]}
\newcommand{\Prob}{\ensuremath{\operatorname{\mathbf{Pr}}}}

\newcommand{\Probc}[1]{\mathbf{Pr}\left[#1\right]}
\newcommand{\sigmafree}{\sigma_{\text{out}}}

\newcommand{\sigmastable}{\sigma_{\text{in}}}

\newcommand{\maxd}{10}

\usepackage[normalem]{ulem}               
\newcommand\redout{\bgroup\markoverwith
{\textcolor{red}{\rule[0.5ex]{2pt}{0.8pt}}}\ULon}

\begin{document}

\maketitle

\begin{abstract}
We consider self-stabilizing algorithms to compute a Maximal Independent Set (MIS) in the extremely weak \emph{beeping} communication model.
The model consists of an anonymous network with synchronous rounds. 
In each round, each vertex can optionally transmit a signal to all its neighbors (beep). 
After the transmission of a signal, each vertex can only differentiate between no signal received, or at least one signal received. 
We also consider an extension of this model where vertices can transmit signals through two distinguishable beeping channels.
We assume that vertices have some knowledge about the topology of the network.

We revisit the not self-stabilizing algorithm proposed by Jeavons, Scott, and Xu (2013), which computes an MIS in the beeping model. 
We enhance this algorithm to be self-stabilizing, and explore three different variants, which differ in the knowledge about the topology available to the vertices and the number of beeping channels.
In the first variant, every vertex knows an upper bound on the maximum degree $\Delta$ of the graph. For this case, we prove that the proposed self-stabilizing version maintains the same run-time as the original algorithm, i.e., it stabilizes after $O(\log n)$ rounds w.h.p.\ on any $n$-vertex graph.
In the second variant, each vertex only knows an upper bound on its own degree. For this case, we prove that the algorithm stabilizes after $O(\log n\cdot \log \log n)$ rounds on any $n$-vertex graph, w.h.p. In the third variant, we consider the model with two beeping channels, where every vertex knows an upper bound of the maximum degree of the nodes in the $1$-hop neighborhood. We prove that this variant stabilizes w.h.p.\ after $O(\log n)$ rounds.
\end{abstract} 

\section{Introduction}

The Maximal Independent Set (MIS) problem has a central role in the areas of parallel and distributed computing. In a graph $G=(V,E)$, an MIS is a subset of vertices $ I \subseteq V$ where no two vertices in $I$ are adjacent, and it is maximal with respect to inclusion. Recognized for its importance in the field of distributed computing since the early 1980s \cite{Luby86,AlonBI86}, the computation of an MIS serves as a foundational subroutine in various algorithms in wireless networks, routing, and clustering \cite{Peleg2000}. The interest in the MIS problem has recently extended to biological networks, with observations of processes similar to the MIS elections in the development of the fly's nervous system \cite{AfekABHBB11}.

While distributed MIS algorithms are well-explored in the standard synchronous message-passing models like LOCAL, CONGEST, and CONGESTED-CLIQUE 
\cite{Peleg2000,Linial1987,Lotker2003,Ghaffari16,GhaffariGR21,Faour0GKR23,BalliuBHORS21,Ghaffari19},
recently the MIS selection was considered also within weaker communication frameworks \cite{MoscibrodaW05a,AfekABCHK13,EmekK21}. 
Indeed, novel distributed communication models, inspired by scenarios in biological cellular networks, wireless sensor networks and networks with sub-microprocessor devices, were defined. 
The Stone Age model, introduced by Emek and Wattenhofer, provides an abstraction of a network of randomized finite state machines that communicate with their neighbors using a fixed message alphabet based on a weak communication scheme \cite{EmekW13}. 
Another related model, which is the one we consider in this paper, is the full-duplex beeping model\footnote{This model is also called the beeping model with collision detection.}, where a network of anonymous processors and synchronous rounds is considered \cite{CornejoK10}.  
In each round, each vertex has the option to either broadcast a signal – a beep – to all its neighbors or to remain silent. Subsequently,  each vertex can determine whether it received any signals or if all its neighbors remained silent. This does not allow a vertex to differentiate which vertex emitted the signal, nor the number of signals received.
We notice that a variation of this model can be defined where, instead of a single type of signal, a constant number of distinct signals exist, and the vertices can distinguish between the types of signals received.
The beeping model finds motivation in scenarios such as wireless sensor networks or biological systems, where organisms can only detect proteins transmitted by neighboring entities \cite{AfekABCHK13}.
The problem of computing an MIS was already considered in the full-duplex beeping model \cite{JeavonsS016,Ghaffari17,AfekABCHK13} and in the Stone Age model \cite{emek2020,EmekK21,EmekW13}.

In both biological and wireless systems, another notable trait is their capability for self-recovery. This ability is also essential in distributed and large-scale systems, which must be able to effectively manage faults.  Self-stabilizing algorithms are designed to ensure that systems can recover from any state and eventually stabilize into a valid state,  maintaining stability as long as faults are absent \cite{Dijkstra74,Dolev2000}. Indeed, self-stabilizing algorithms are
guaranteed to converge from any initial configuration. However, only a few self-stabilizing MIS algorithms have been proposed for the aforementioned weak communication models \cite{AfekABCHK13,EmekK21,Giakkoupis2023}.
In the full-duplex beeping model, Afek et al.\ in \cite{AfekABCHK13} introduced a self-stabilizing algorithm that converges to an MIS in $O(\log^2 N \log n)$ rounds with high probability (w.h.p.), if all vertices know an upper bound $N$ on the network's size  $n$. 
They also established a polynomial lower bound for the MIS in a similar model. This model includes an adversary able to select the wake-up time slots for the vertices. Because of the presence of the adversary, the lower bound of \cite{AfekABCHK13} is not applicable in the setting of this paper.
In the full-duplex beeping model, a constant-state algorithm was proposed in \cite{Giakkoupis2023}, which stabilizes in poly-logarithmic rounds w.h.p., albeit being efficient only for some graph families. 
Meanwhile, Emek et al.\ \cite{EmekK21} devised an 
algorithm for a simplified version of the Stone Age model that is slightly stronger than the beeping communication model, which stabilizes in $O((D+\log n)\log n) $ rounds w.h.p. on any $D$-bounded diameter graph, where $D$ is considered a fixed parameter. 
However, in this context, it would be desirable to relinquish the assumption that vertices possess global information about the network’s structure.

Algorithms that do not require any knowledge of the network's topology were also proposed for the beeping model, but they strongly rely on the assumption that, at the beginning of the algorithm, the vertices are in the same fixed initial state, and hence they are not self-stabilizing. One algorithm was proposed by Afek et al.\ \cite{AfekABCHK13} for the full-duplex beeping model, which stabilizes in $O(\log^2 n)$ rounds w.h.p., without requiring vertex knowledge of the network's topology. 
Later, Jeavons et al.\ \cite{JeavonsS016} improved this result by proposing an algorithm for the same model, capable of computing an MIS in any $n$-vertex graph in $O(\log n)$ rounds w.h.p., without requiring any vertex knowledge\footnote{Ghaffari provided a refined analysis for Jeavons at al.'s algorithm in \cite{Ghaffari17}.}. 
Notice that these algorithms are not self-stabilizing because they also rely on the presence of phases of two rounds, implying a synchronization of the vertices modulo two.

\subsection{Our Contribution}

In this paper, we propose a self-stabilizing algorithm for computing the MIS in the full-duplex beeping model, aiming for a stabilization time of $O(\log n)$ with minimal vertex knowledge about network topology. 

We consider the standard fault model, used in most self-stabilizing algorithms \cite{Dolev2000}, where the state of each node is stored in RAM and data in RAM can be corrupted by transients faults (e.g., external events), while the code is stored in ROM and cannot be corrupted. We consider a fault-free execution after a RAM corruption. An algorithm $\mathcal{A}$ is self-stabilizing with termination time $T$ if, after a transient fault within $T$ fault-free steps, it reaches a legal state. This is equivalent to asking that the algorithm $\mathcal{A}$ reaches a legal level after $T$ fault-free steps, starting from an arbitrary state, i.e., without a fixed initialization.

The starting point of our work is Jeavons' algorithm in \cite{JeavonsS016}, which is non-self-stabilizing and converges within $O(\log n)$ rounds. We propose two variants that achieve self-stabilization and efficiency across all graph sizes. Our algorithms rely on each vertex's ability to compute a quantity $\ellmax(v)$, which may require access to some information, such as the maximum degree of the graph. 
The first variant assumes that vertices know an upper bound on the maximum degree $\Delta$ and stabilizes in $O(\log n)$ rounds, while the second variant assumes that each vertex knows an upper bound on its own degree and stabilizes in $O(\log n \log \log n)$ rounds. Additionally, we present a third algorithm for the extended beeping model with two channels, stabilizing in $O (\log n)$ time if vertices know an upper bound on the maximum degree among the 1-hop neighborhood.
In summary, our contributions yield three algorithms for computing MIS in the beeping model, each highlighting different scenarios based on varying levels of vertex knowledge and beeping channels. 
Formally, we prove the following theorem.

\begin{theorem}
    Let $G$ be a $n$-vertex graph.
    \begin{enumerate}
        \item
        If each vertex knows the same upper bound on the maximum degree of $G$, which is at most polynomial in $n$, 
        then an MIS can be computed in the beeping model, in a self-stabilizing manner, within $O(\log n)$ rounds w.h.p.
        
        \item If each vertex knows an upper bound  on its own degree, which is at most polynomial in $n$, then an MIS can be computed in the beeping model, in a self-stabilizing manner, within $O(\log n \log \log n)$ rounds w.h.p.
        
        \item If each vertex knows an upper bound on the maximum degree of all vertices in its $1$-hop neighborhood, which is at most polynomial in $n$, then an MIS can be computed, in the beeping model with two channels, in a self-stabilizing manner, within $O(\log n)$ rounds w.h.p. 
    \end{enumerate}
\end{theorem}
It remains an open question whether a fast, self-stabilizing algorithm computing an MIS in the beeping model can be designed so that no information about the network topology is required to be known by the vertices.

\section{The Algorithm}
\label{sec:motivation}
We assume the full-duplex beeping communication model and the starting point for our algorithm is the beeping, randomized algorithm of Jeavons et al.\ in \cite{JeavonsS016}.
Each vertex $v$ is associated with an adaptive probability $p_t(v)$ of beeping in round $t$, and the algorithm works in phases, each consisting of two rounds. In the first round of each
phase, each vertex $v$ beeps with probability $p_t(v)$ and, if $v$ beeps
and all its neighbors are silent, then $v$ joins the MIS. In the second round of each phase, vertices that joined the MIS beep and neighboring
vertices hearing a beep become non-MIS vertices. Then, the newly joined MIS and
non-MIS vertices remain silent for the rest of the algorithm. The crucial
point leading to a $O(\log n)$ global round complexity with high
probability, is that active vertices adapt in each phase the beeping
probability, initially $p_1(v) = 1/2$ for each vertex $v$. The value of
$p_{t+1}(v)$ is decreased whenever neighboring vertices beep and is
increased otherwise. In particular $p_{t+1}(v)=p_t(v)/2$ in the former
case and $p_{t+1}(v)= \min \{2p_t(v),1/2\}$ otherwise.
The rationale of this behavior is twofold: to reduce the probability of neighboring
vertices attempting to concurrently join the MIS, and to increase the
probability of making an attempt to join the MIS in case of no
concurrent attempts to do so.

This algorithm is not self-stabilizing for two reasons. First, it works just if at the beginning of the algorithm the probability of beeping of each vertex $v$ is $p_1(v)=1/2$, and the analysis of the convergence time relies on that. Second, the presence of phases with two rounds requires that the vertices are synchronized modulo two. These reasons are also the main obstacle to making it
self-stabilizing.
Moreover, in self-stabilizing algorithms, vertices must be able to detect errors, e.g., when a fault forces a vertex to change its state from MIS to non-MIS, and hence stable vertices cannot be silent after they stabilized.

In order to design a self-stabilizing MIS algorithm for the full-duplex beeping
model, achieving a $O(\log n)$ global round complexity w.h.p., we dispense with the idea of phases and we change the
details of updating the beeping probabilities $p_t(v)$ to overcome
the mentioned issues. While keeping the idea of increasing and decreasing the
beeping probability depending on whether a beep was received, we refine
this behavior in a significant way. As before, when a vertex $v$ beeps
while hearing no beeps at the same time it attempts to join the MIS.
To signal this to neighboring vertices, vertex $v$ keeps beeping, i.e., it sets its
beeping probability $p_t(v)$ to $1$. If such a vertex hears a beep in one of the
following rounds, it does not immediately give up its attempt to join the MIS, but it keeps beeping with probability $1$ for some fixed number rounds. Only after hearing a beep in a certain number of rounds, the vertex changes
its behavior back to halving its beeping probability in every round it hears a beep. Furthermore, if the beeping probability decreases over a fixed threshold, the vertex sets its beeping probability to $0$ and stops beeping.
The complete code is shown in \cref{alg:jsx-ss}.

\begin{algorithm}[ht]
    
    \state{$\ell \in\{-\ellmax(v),\ldots, \ellmax(v)\}$}
    
    \BlankLine
    
    \InEachRound{$t=1,2,\ldots$}{

        \uIf{$\ell < \ellmax(v)$}
            {\label{jsx:rand}
            $beep \,\gets\, \true$ with probability 
            $\min\left\{2^{-\ell},1\right\}$            and $beep \,\gets\, \false$ otherwise\;} 
        \lElse{$beep \,\gets\, \false$}

        \BlankLine

        \lIf{$beep$}{send signal to all neighbors \label{jsx:send}}
        receive any signals sent by neighbors\;
        \BlankLine

        \BlankLine

        \uIf {any signal received}
            {\label{jsx:incr}
            $\ell \,\gets\, \min\{\ell+1,\,\ellmax(v)\}$\;}
        \uElseIf {$beep$}
            {\label{jsx:reset}
            $\ell \,\gets\, -\ellmax(v)$\;}
        \lElse
            {\label{jsx:decr}%
            $\ell \,\gets\, \max\{\ell-1, 1\}$}
    }
    
    \caption{Self-stabilizing version of Jeavons, Scott and Xu's algorithm \cite{JeavonsS016}}
    \label{alg:jsx-ss}
\end{algorithm}

To implement the described behavior, each vertex maintains an integral state
variable $\ell$, which we call \emph{level}. The value of $\ell$ for vertex $v$ is in the range
$-\ellmax(v),\ldots, \ellmax(v)$, where $\ellmax(v)$ is a fixed value that depends on the vertex's knowledge of some graph parameters.  
 We will see that this value has a strong influence on the analysis of the stabilization
time. The value of $\ell_t(v)$ of vertex $v$ in round $t$ implies the beeping probability $p_t(v)$ of $v$ similar to an activation function in an artificial neural network (see \cref{fig:activation}). 
As long as $\ell_t(v) \le 0$ vertex $v$ beeps and $p_t(v)=1$, if $\ell_t(v)=\ellmax(v)$ it stops beeping and $p_t(v)=0$, otherwise $p_t(v)=2^{-\ell_t(v)}$.

\begin{figure}[h]
  \begin{center}
  \includegraphics[scale=0.95]{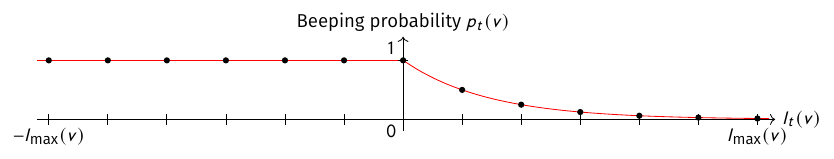}
\end{center}
  \caption{Beeping probability $p_t(v)$ of $v$ based on value of $\ell_t(v)$.}\label{fig:activation}
\end{figure}

In each round $t$ each vertex $v$ updates the value of $\ell_t(v)$ as
follows. If $v$ hears a beep then its level increases:
$\ell_{t+1}(v) = \min \{\ell_t(v) +1,\ellmax(v)\}$. Otherwise,
$\ell_{t+1}(v) = \max \{\ell_t(v) -1,1\}$ unless $v$ was beeping in
round $t$, in this case $\ell_{t+1}(v) = -\ellmax(v)$. Note that the only way the level of a vertex $v$ can decrease below $0$ is if $v$ beeps without beeping neighbors. 
We observe that \cref{alg:jsx-ss} is self-stabilizing if its convergence is guaranteed for every initial value of the levels. 

The update rules of the algorithm guarantee that, once the level's value of a vertex $v$ is $-\ellmax(v)$ and each of $v$'s
neighbors $w$ has level's value $\ellmax(w)$, then $v$ is such that $p_t(v)=1$ and all the neighbors $u$ of $v$ are such that $p_t(u)=0$. This guarantees that $v$ and its neighbors will not change their status as long as no faults occur, and hence they are stable. In this case, $v$ will be
a MIS vertex and the neighbors become non-MIS vertices. Also, this strategy
allows all vertices to detect faults and react accordingly. But foremost,
it allows to determine the stabilization time.

The result and the analysis of the algorithm depend on the values $\ellmax(v)$ of each vertex $v$, which in turn depends on the knowledge available to each vertex $v$. We state the detailed results in the following theorems, and notice that we denote with $\deg_2(v)= \max_{u \in N(v)\cup\{v\}}\deg(u)$ the maximum degree in the $1$-hop neighborhood of $v$.

\begin{theorem}
\label{thm:jeavons} For any $n$-vertex graph $G$, 
\cref{alg:jsx-ss} computes an MIS, starting from an arbitrary configuration, within $O(\log n)$ rounds w.h.p., provided that $\ellmax(v) = \ellmax \in [\log \Delta +c_1, c_2 \log n]$  for each vertex $v$ and constants $c_1 \geq 15$ and $c_2 > 0$.
\end{theorem}

\begin{theorem}
\label{thm:jeavons_degree} For any $n$-vertex graph $G$,
\cref{alg:jsx-ss} computes an MIS, starting from an arbitrary configuration, within $O(\log n \cdot \log \log n)$ rounds w.h.p., provided that $\ellmax(v) \in [2 \log \deg(v) +c_1, c_2 \log n]$  for each vertex $v$ and constants $c_1 \geq 30$ and $c_2 > 0$.
\end{theorem}

\begin{corollary}
There exists a variant of \cref{alg:jsx-ss} for the beeping model with two beeping channels such that, for any $n$-vertex graph $G$, it computes an MIS, starting from an arbitrary configuration, within $O(\log n)$ rounds w.h.p., provided that $\ellmax(v) \in [2\log \deg_2(v)+c_1, c_2\log n]$,  for each vertex $v$ and any constants $c_1 \geq 15$ and $c_2 >0$.
\label{thm:two_beeping_channels}
\end{corollary}

To execute \cref{alg:jsx-ss}, each vertex $v$ only needs to know the value of $\ell_{max}(v)$. As stated in the three results above, in order to get the time bounds, the value of $\ellmax(v)$ must be in $O(\log n)$ for each $v$. We remark that to satisfy this requirement it is unnecessary that the value $n$ is known by the vertices. If, for example, $\ell_{max} = \log \Delta + c_1$, then the requirement of \cref{thm:jeavons} is satisfied, and this only requires each node to know a loose upper bound on $\Delta$.

\subparagraph*{Roadmap.}

The rest of the paper is organized as follows. \cref{sec:stabilization_jeavons} contains notations, preliminary definitions, the statement of two key lemmas, \cref{lem:lower_bound_platinum,lem:char_platinum}, and an analysis outline. In \cref{sec:warm_up} we give the proof of \cref{thm:jeavons}, and in \cref{sec:thm:jeavons_degree} the proof of \cref{thm:jeavons_degree}. The proofs of key \cref{lem:lower_bound_platinum,lem:char_platinum}
can be found in \cref{sec:proof_main_lemmas}. 
The description of the algorithm using two beeping channels and its analysis (the proof of \cref{thm:two_beeping_channels}) are deferred to \cref{sec:two_beeping_channels}.
We conclude in \cref{sec:conclusion} with a summary and some open problems.

\section{Definitions and Analysis Outline}
\label{sec:stabilization_jeavons}

Let $G=(V,E)$ be a graph with $n$ vertices. For each vertex $v \in V$, $N(v)$ denotes the set of $v$'s neighbors in $G$, and $\deg(v)=|N(v)|$ is the degree of $v$. Also, $N^+(v)  = N(v) \cup \{v\}$ is the set of $v$'s neighbors and $v$ itself. Let $\deg_2(v) = \max_{u \in N^+(v)}\deg(u)$ the maximum degree of all the vertices in $N^+(v)$.

We introduce a few random variables that are used to describe the random process generated by the execution of \cref{alg:jsx-ss}. If we denote with $\ell_t(v)$ the level of vertex $v \in V$ at the beginning of round $t \geq 1$, the random execution of the algorithm at time $t$ depends only on the values $\{\ell_t(v)\}_{v \in V}$. We denote with $\mathcal{F}_t$ the filtration of the process until step $t$, which in particular gives us the values $\{\ell_t(v)\}_{v \in V}$. 

We notice that in \cref{alg:jsx-ss} a vertex $v \in V$ is \emph{stable} and permanently added to the  MIS prior to round $t$ if $\ell_t(v)=-\ellmax(v)$ and, for all $u \in N(v)$, $\ell_t(u)=\ellmax(u)$. Hence, if we define \[\mu_t(v)=\min_{u \in N(v)}\frac{\ell_t(u)}{\ellmax(u)},\]
which has value in $[-1,1]$, we have that the set of vertices that have been added to the final MIS set before round $t$ is defined by
\[
    \mis_t 
    = 
    \{v\in V \colon \ell_t(v) = -\ellmax(v) \,\land\, \mu_t(v) = 1\}.
\]
Moreover, the set of all stable vertices at the beginning of round $t$ consists of the vertices in the MIS and their neighbors, so we define $S_t = I_t \cup N(I_t)$. We notice that the set of stable vertices is increasing in $t$, i.e., for each $t \geq 1$ we have that $S_{t}\subseteq S_{t+1}$.
For any vertex $v \in V$, we denote with $p_t(v)$ the probability that $v$ beeps during round $t$, which is
\[
    p_t(v) =
    \begin{cases}
        1 & \text{if }\ell_t(v)\le 0 \\
        2^{-\ell_t(v)} &\text{if }0 <\ell_t(v)<\ellmax(v)\\
        0 &\text{if }\ell_t(v)=\ellmax(v).
    \end{cases}  
\]
We also denote with $b_t(v)$ a Bernoulli random variable which takes value $1$ if $v$ beeps in round $t$, i.e., $\Expc{b_t(v)}=p_t(v)$. We define $B_t(v) = \sum_{u\in N(v)} b_t(u)$ as the number of $v$'s neighbors that beep in round $t$ and $d_t(v)=\Expc{B_t(v)}=\sum_{u \in N(v)}p_t(u)$ as the expected number of beeping neighbors of $v$ in round $t$. Note that if $B_t(v)=0$ then $\mu_t(u)> 0$ for all neighbors $u$ of $v$. 

\begin{lemma}
\label{fact:no_adjacent_negative}
   Let $t > \max_{w \in V}\ellmax(w)$. Then $\ell_t(v) > 0$ or $\mu_t(v) > 0$ for any $v \in V$.
\end{lemma}

\begin{proof}[Proof of \cref{fact:no_adjacent_negative}]
    Let $t_0$ be the first round such that $\ell_{t_0}(v) > 0$ or $\mu_{t_0}(v) > 0$. First, we will prove that this condition continues to hold for all rounds $t\ge t_0$. Then, we will prove that $t_0 \leq \max_{w \in V}\ellmax(w)+1$.

    Consider any round $t \geq t_0$ and assume that $\ell_t(v)>0$ or $\mu_t(v)>0$. This implies that $\mu_{t+1}(v)>0$ or $\ell_{t+1}(v)>0$. Indeed, assume that $\mu_{t}(v) \le 0$. Then $\ell_{t}(v) > 0$ and at least one neighbor of $v$ beeps in round $t$. Thus, $\ell_{t+1}(v)=\min\left\{\ell_{t}(v)+1,\ellmax(v)\right\}\ge \ell_{t}(v)>0$, i.e., the condition of the lemma holds in round $t+1$.     
    Next consider the case that $\mu_{t}(v) > 0$. If $v$ beeps in round $t$ then all neighbors increase their value for $\ell$, i.e., $\mu_{t+1}(v) \ge \mu_{t}(v)> 0$. If $v$ does not beep in round $t$  then $\ell_{t}(v)>0$. Indeed, if no neighbor of $v$ beeps then   
    $\ell_{t+1}(v)=\max\left\{\ell_{t}(v)-1,1\right\}>0$, and if at least one neighbor of $v$ beeps then $\ell_{t+1}(v)=\min\left\{\ell_{t}(v)+1,\ellmax(v)\right\}\ge \ell_{t}(v) >0$, i.e., the condition of the lemma holds in round $t+1$. 

    Assume that $\ell_{0}(v)\le 0$ and $\mu_{0}(v)\le 0$. Then, in the first round all vertices in $N^+(v)$ beep. Hence, all these vertices increment their level by $1$, i.e., $\ell_{1}(v)=\ell_{0}(v)+1$ and $\mu_1(v) = \min_{u \in N(v)}\frac{\ell_0(u)+1}{\ellmax(u)}$. Since $-\ellmax(u) \le \ell_{0}(u)$ for all vertices $u\in V$, there exists $t_0\le \max_{u \in N^+(v)}\ellmax(u)+1$, such that $\ell_{t_0}(v) > 0$ or $\mu_{t_0}(v) > 0$. This completes the proof.
\end{proof}

\cref{fact:no_adjacent_negative} implies that in order to prove that our algorithm stabilizes within $O(\log n)$ rounds we can assume that $\ell_{t}(v) > 0$ or $\mu_{t}(v) > 0$ for all rounds $t\ge 0$. This is because $\max_{w \in V}\ellmax(w) \in O(\log n)$. Hence, we can ignore the initial $\max_{w \in V} \ellmax(w)$ rounds and start our analysis after those rounds. In particular, $\ell_t(u)\leq 0$ implies $\mu_t(u)>0$.

We define a vertex to be \emph{prominent} if it has negative or zero level, and a round to be \emph{platinum} for some vertex $v$ if some of $v$'s neighbors is prominent.

\begin{definition}[Prominent Vertices and Platinum Rounds]
A vertex $v \in V$ is \emph{prominent} in round $t$ if $\ell_t(v) \leq 0$. The set of prominent vertices in round $t$ is denoted with $\Promi_t$. Moreover, we say that round $t$ is \emph{a platinum round of vertex $v$} if  $N^+(v)$ contains a prominent vertex $u$, i.e., $u \in N^+(v) \cap \Promi_t$. We denote with $P_{t,k}(v)$ the number of platinum rounds of vertex $v$ during  rounds $\{t, \dots, t+k\}$.
 \label{def:platinum}
\end{definition}

Clearly, $I_t\subseteq \Promi_t$. We notice that, since we assume $t > \max_{w \in V}\ellmax(w)$, then \cref{fact:no_adjacent_negative} implies that for each platinum round $t$ of $v$ there exists $u\in N^+(v)$ such that $\ell_t(u) \leq 0$ and $\mu_t(u)>0$,  i.e., the probability that none of $u$'s neighbors beeps in round $t$ is positive. 
Remember, the only possibility for the level of vertex $u$ to become less or equal to $0$ is when $u$ beeps while no neighbor of $u$ is beeping. 
This directly leads to the next lemma.
\begin{lemma}    \label{fact:no_beeping_neighbors}
If $t>\max_{w \in V}\ellmax(w)$ is a platinum round for vertex $v$ there exists a vertex $u \in N^+(v)$ and a round $t'$ with $t-\ellmax(u)\leq t'\leq t$ in which $u$ was beeping without beeping neighbors and $\ell_{t'+1}(u)=-\ellmax(u)$.
\end{lemma}

We define, for any $v \in V$ and $t \geq 1$, the quantities
\[\eta_t(v) = \sum_{u \in N(v) \setminus S_t} 2^{-\ellmax(u)} \quad \text{ and }\quad \eta'_t(v)= \sum_{\substack{u \in N(v)\setminus S_t:\\ \ellmax(u)> \ellmax(v)}}2^{-\ellmax(v)}.\]
For the moment, the definitions of $\eta_t(v)$ and $\eta'_t(v)$ are rather technical, but they will be used  to upper bound the value of $d_{t+1}(v)$.  We notice  that $\eta_t(v)$ and $\eta_t'(v)$ are both decreasing in $t$, since $S_t \subseteq S_{t+1}$.

The following two lemmas are the key to prove  \cref{thm:jeavons,thm:jeavons_degree,thm:two_beeping_channels}, their proofs are deferred to \cref{sec:proof_main_lemmas}.
For a fixed $v \in V$ the next lemma tells us how many rounds we have to wait in order to have a platinum round of $v$.

\begin{lemma}[Lower Bound on Platinum Rounds]
\label{lem:lower_bound_platinum}
Assume that $\ellmax(w) \geq \log \deg(w)+4$ for all $w \in V$.
Consider a vertex $v \in V$ and a round $t > \max_{w\in V} \ellmax(w)$ such that $t$ is not a platinum round of $v$, and $\eta_t(v) \leq 0.0001$. Let 
$\tau^{(v)}(t) = min\{m\geq 0: P_{t,m}(v)\geq 1\}$. 
Then \[\Probc{\tau^{(v)}(t) \geq k\mid \mathcal{F}_t} \leq e^{-\gamma k},\]  for $\gamma=e^{-30}$ and any $k \geq 2 \gamma^{-1} \ellmax(v)$.
\end{lemma}

We notice that, if $\ellmax(w)$ is constant over all vertices $w \in V$, i.e., $\ellmax(w)=\ellmax$ for every vertex $w \in V$, then the existence of a platinum round $t$ of $v$ such that $t > \ellmax$ is by \cref{fact:no_beeping_neighbors} sufficient to guarantee that $v$ will be stable at the latest in round $t+\ellmax$. Indeed, \cref{fact:no_beeping_neighbors} implies the existence of a round $1 \leq t' \leq t$ and a vertex $u \in N^+(v)$ such that $u$ was beeping in round $t'$ without beeping neighbors, and so $\ell_{t'+1}(u)=-\ellmax$ and $\mu_{t'+1}(u)>0$. This implies that $u$ beeps in the following $\ellmax$ rounds, during which all neighbors of $u$ will increase their level until they reach maximum level $\ellmax$. This implies that $u$ is such that $\ell_{t+\ellmax}(u)\leq 0$ and $\mu_{t+\ellmax}(u)=1$, and hence $u,v \in S_{t+\ellmax}$ and $u \in I_{t+\ellmax}$. 

However, when $\ellmax(w)$ is not constant, the analysis becomes considerably more complicated, since the existence of a platinum round of $v$ does not  necessarily imply the subsequent stabilization of $v$. 
Consider now some round $t >\max_{w \in V} \ellmax(w)$ which is platinum for $v$, and let $u \in N^+(v)$ be a prominent vertex. After round $t$, two things may happen:
\begin{enumerate}[(i)]
    \item In some round $t+m$ with $m \geq 1$, $u$ is no longer prominent, and hence $u \not \in I_{t+m}$ and $u,v$ may not be stable in round ${t+m}$;
    \item In some round $t+m$ with $m \geq 0$, vertex $u$ is prominent and all its neighbors have reached the maximum level, i.e., $\mu_{t+m}(u)=1$, and so $u \in I_{t+m}$ and $u,v \in S_{t+m}$.
\end{enumerate}
In the next lemma, we characterize the distribution of rounds $t$ for the above two cases. Let
\begin{gather*}
    \sigmafree^{(u)}(t) = \min\{m \geq 0:u \not \in \Promi_{t+m}\}
    \qquad
    \\
    \sigmastable^{(u)}(t) = \min\{m\geq 0: u \in I_{t+m}\}
    \\
    \sigma^{(u)}(t) = \min\{\sigmafree^{(u)}(t),\ \sigmastable^{(u)}(t)\}.
\end{gather*}

\begin{lemma}[Stopping Times for Platinum Rounds]
Assume that $\ellmax(w) \geq \log \deg(w) + 4$ for all $w \in  V$. 
Consider a round $t > \max_{w \in V}\ellmax(w)$ and a vertex $u \in \Promi_t\setminus S_t$. Then
\begin{enumerate}[(a)]
    \item $\Probc{\sigma^{(u)}(t) = \sigmastable^{(u)}(t) \wedge \sigma^{(u)}(t) < \max_{w \in N(u)}\ellmax(w) \mid \mathcal{F}_t}\geq 3^{-\eta'_t(u)}$;
    \item $\Probc{\sigma^{(u)}(t)=\sigmafree(t) \wedge \sigma^{(u)}(t)>\ellmax(u)+x \mid \mathcal{F}_t} \leq \eta'_t(u) 2^{-x}$ for any $x \geq 0$.
\end{enumerate}
\label{lem:char_platinum}
\end{lemma}

\subsection{Analysis Overview}

We first give an overview of the proofs of \cref{lem:lower_bound_platinum,lem:char_platinum}, and then we will see how to use these results to prove \cref{thm:jeavons,thm:jeavons_degree,thm:two_beeping_channels}. 
The proof of \cref{lem:lower_bound_platinum} has as a starting point the proof in \cite{Ghaffari17}, but then it develops differently.
First, as in \cite{Ghaffari17}, we define a further type of round called \emph{golden round}, which are rounds having constant probability of becoming platinum in the subsequent round. We prove that, for any vertex $v$ in any fixed interval of rounds  of length $k=\Omega(\ellmax(v))$, we have a constant fraction of golden rounds with probability at least $1-e^{-\Omega(k)}$, conditioned on the absence of platinum rounds during that time interval. To prove the latter, as in \cite{Ghaffari17}, we analyze the development of the function $d_t(v)$ -- the expected number of beeping neighbors of
$v$ in round $t$ -- during this time span. Note that platinum rounds and the conditioning were not considered in \cite{Ghaffari17} and are essential in our proof and setting.

The proof of \cref{lem:char_platinum} relies on \cref{fact:no_beeping_neighbors}. Assuming that $u$ is prominent at time $t$, we characterize the probabilities with which, after round $t$, $u$ reaches again a positive level or stabilizes. From \cref{fact:no_beeping_neighbors} there exists a round $t-\ellmax(u) \leq t' \leq t$ where we have that $\ell_{t'+1}(u)=-\ellmax(u)$ and then trivially $d_{t'+1}(u)=\sum_{w \in N(u)}p_{t'+1}(u) \leq \deg(u)$.
Then, in the subsequent $\ellmax(u)$ rounds, vertex $u$ keeps beeping regardless the behavior of the vertices in $N(u)$. Hence, $\ell_{t'+1+\ellmax(u)}(w)=\min \{\ellmax(w), \ell_{t'+1}(w) +\ellmax(u)\}$ for each $w\in N(u)$ and thus, $p_{t'+1+\ellmax(u)}(w) \leq 2^{-\ellmax(u)}$ if $\ell_{t'+1+\ellmax(u)}(w) \not=\ellmax(w)$. This implies that \[d_{t'+1+\ellmax(u)}(u) =\sum_{w\in N(u)} p_{t'+1+\ellmax(u)}(w)\leq\sum_{\substack{w \in N(u)\setminus S_{t'+1}:\\ \ellmax(w)> \ellmax(u)}}2^{-\ellmax(u)}\leq \eta_{t'+1}'(u).\] We will see that this implies that the vertices in $N(u)$ will reach their maximal level with probability at least $3^{-\eta'_{t+1}(u)}$, and so in this case the platinum round leads to the stabilization of $u$. 
On the other hand, part (b) of the lemma follows from the observation that, after the first $\ellmax(u)$ rounds after $t$, the probability that some vertex in $N(u)$ beeps decreases in each round by a constant factor.

\cref{thm:jeavons,thm:two_beeping_channels} follow from the observation, already stated above, that if $\ellmax(w)$ is constant over $w$ then, for each vertex $v$, one platinum round is sufficient to guarantee the stabilization of $v$. Moreover,  the choices of $\ellmax(w)$ specified in the theorems guarantee that $\eta_t(v) \leq 0.0001$ for every $v$ and $t\geq 1$, and so \cref{lem:lower_bound_platinum} can always be used for each non-platinum round $t$, and implies that we have to wait at most $O(\log n)$ rounds to have a platinum round for each vertex $v$ w.h.p., that in turns imply stabilization.

The proof of \cref{thm:jeavons_degree} is considerably harder. In this case, we can have several sequences of consecutive platinum rounds, intermittent by sequences of consecutive non-platinum rounds, until we reach a platinum round leading to the stabilization of the vertex. The analysis relies on two main parts:
\begin{enumerate}[(1)]
    \item We split the vertices in $O(\log \log n)$ sets $V_i$. Before analyzing the stabilization of a vertex $v \in V_i$, we wait for round  $T_i$ in which all vertices in $\cup_{j<i}V_j$ have stabilized. The sets $V_i$ are defined according to the values $\ellmax(v)$ of the vertices. According to the definition of $T_i$, we can apply, for each round  $t \geq T_i$,  \cref{lem:lower_bound_platinum,lem:char_platinum} to vertices in $V_i$.
    \item We then prove that, after round $T_i$, each vertex $v \in V_i$ stabilize in $O(\log n)$ additional rounds w.h.p. The analysis of the latter statement relies on  \cref{lem:lower_bound_platinum,lem:char_platinum}, which characterize the lengths of three times intervals: that of the non-platinum rounds, of the platinum rounds, and that of the number of platinum rounds not leading to the stabilization of vertex $v$.
\end{enumerate}

\section{Knowledge of  Maximum Degree \texorpdfstring{$\Delta$}{Delta} (Proof of
\texorpdfstring{\cref{thm:jeavons}}{Theorem 2})}
\label{sec:warm_up}

The following proof is a warm-up for the general case. It is directly implied by \cref{lem:lower_bound_platinum} and the choice of $\ellmax(v)$.

\begin{proof}[Proof of \cref{thm:jeavons}]
    As already mentioned, since $\ellmax(v)$ is defined independently of $v$, each vertex $v$ just requires a single platinum round to become stable in at most $\ellmax$ rounds. Indeed, for each $v \in V$ and each $t \geq 1$, 
    \[\eta_t(v) \le \sum_{u \in N(v)} 2^{-\log\Delta - 15} \le 2^{-15} \leq 0.0001 \quad \text{and} \quad \eta_t'(v)=0.\]
    This implies that, if $t = 2\ellmax$ and $v \in V$, we have  $\eta_t(v) \leq 0.0001$.
    Hence,  by \cref{lem:lower_bound_platinum}, if we take $m = 2\gamma^{-1}\log n$ (where $\gamma$ is defined in \cref{lem:lower_bound_platinum}), we have that $\Probc{\tau^{(v)}(t) \leq m \mid \mathcal{F}_t} \geq 1- 1/n^2$, and so $P_{t,m}(v) \geq 1$ with probability at least $1-1/n^2$.  Then, from \cref{lem:char_platinum}(a), and since $\eta_{t+m}'(v)=0$, we have that, given $\mathcal{F}_{m+t}$, the vertex $v$ is stable after at most $\ellmax$ rounds with probability $1$. Hence, vertex $v$ is stable with probability $1-1/n^2$ after $t+m+\ellmax$ rounds, and since $\ellmax=O(\log n)$ we have that $t+m+\ellmax = O(\log n)$. The theorem follows from the union bound over all the vertices.
\end{proof}

\section{Knowledge of Own Degree (Proof of \texorpdfstring{\cref{thm:jeavons_degree}}{Theorem 3})}
\label{sec:thm:jeavons_degree}

In this section we prove \cref{thm:jeavons_degree}. First, we prove the following lemma. 

\begin{lemma}
Assume that $\ellmax(w) \geq 2\log \deg(w)$ for every $w \in V$ and that, for some $c=O(1)$, $\max_{w \in V}\ellmax(w) \leq c \log n$.
Consider a vertex $v \in V$ and a round $t > \max_{w \in V} \ellmax(w)$ such that $\eta_t(v) \leq 0.0001$ and $\ellmax(v) \leq 2 \ellmax(u)$ for each $u \in N(v) \setminus S_t$. 
Then, there exists a constant $M=O(1)$ such that $\Probc{v \in S_{t+m}\mid \mathcal{F}_t}\geq 1-1/n^2$, provided $m=M \log n$.
\label{lem:stabilization}
\end{lemma}
\begin{proof}
We fix the execution up to the end of round $t$, so
we do not have to condition probabilities on $\mathcal{F}_t$.
We consider the sequence of rounds (which may also be infinite, with $J=+ \infty$) \[t \leq t + \tau_1 \leq t + \tau_1 + \sigma_1 = m_1 + t \leq \cdots \leq t+m_{J-1}+ \tau_J \leq t+m_{J-1} + \tau_J + \sigma_J =t+m_J,\]
and the corresponding sequence of vertices $v_1,v_2,\dots,v_J \in N^+(v)\setminus S_t$ such that
\begin{enumerate}
    \item $t + m_{i-1} + \tau_i$ is platinum for $v$ and $v_i \in \Promi_{t + m_{i-1} + \tau_i}\cap N^+(v)$ for each $i=1,\dots, J$;
    \item $m_i = m_{i-1} + \tau_i + \sigma_i$ is such that $v_i \not \in \Promi_{t+m_i}$ for each $i=1,\dots J-1$;
    \item $J= \min\{h \geq 1: v_{h} \in I_{m_h + t}\}$, hence $v_J \in I_{t+m_J}$ and $v \in S_{t+m_J}$. If $v$ never stabilizes, we define $J=+\infty$ and the sequence $v_1,v_2, \dots$ has infinite length.
\end{enumerate}
We observe that $\sigma_i$ and $\tau_i$ are defined such that
\begin{equation}
    \tau_i = \tau^{(v_i)}(t+m_{i-1}), \quad \sigma_i = \sigma^{(v_i)}(t+m_{i-1}+\tau_i) \quad \text{ and } \quad \sigma_J = \sigmastable^{(v_J)}(t+m_{J-1}+\tau_J).
    \label{eq:char_stopping_times}
\end{equation}

Consider the following two facts:
\begin{enumerate}[(i)]
    \item  $\sum_{i=1}^J (\sigma_i+\ellmax(v_i))\leq M_1 \log n$  for some $M_1=\Theta(1)$ with probability at least $1-1/n^3$;
    \item Provided that $\sum_{i=1}^J (\sigma_i + \ellmax(v_i)) \leq M_1 \log n$, it holds $\sum_{i=1}^J \tau_i \leq M _2\log n$ for some $M_2 = \Theta(1)$  with probability at least $1-1/n^3$.
\end{enumerate}

The above facts (i) and (ii) prove the lemma. Indeed, if $m = M_1 \log n + M_2 \log n$, we have that
\begin{align*}
    &\Probc{v \not \in S_{t+m}} \leq \Probc{\sum_{i=1}^J(\sigma_i + \tau_i) \geq m} \leq \Probc{\sum_{i=1}^J \sigma_i \geq M_1 \log n \vee \sum_{i=1}^J \tau_i \geq M_2 \log n}
    \\ &\leq \Probc{\sum_{i=1}^J \tau_i \geq M_2 \log n \;\middle|\; \sum_{i=1}^J (\sigma_i + \ellmax(v_i)) \leq M_1\log n} +
    \\ & \quad+ \Probc{\sum_{i=1}^J (\sigma_i + \ellmax(v_i)) \geq M_1 \log n} \leq \frac{2}{n^3}.
\end{align*}
Now we prove (i) and (ii) separately.

\subparagraph{Proof of (i).}
We remark that, in this first step, we are just looking at the randomness of the execution during the time intervals $[t+\tau_i+1, t+\tau_i+\sigma_i]$ for $i=1,\dots, J$.
We notice that
\begin{align}
\notag
&\Probc{\sum_{i=1}^J \sigma_i+ \ellmax(v_i) \geq M_1 \log n} 
\leq \Probc{\sum_{i=1}^J \ellmax(v_i) \geq 7 \log n \vee \sigma_J \geq \max_{w \in V}\ellmax(w)} 
\\ &+ \Probc{\sum_{i=1}^J \ellmax(v_i) \leq 7 \log n \wedge \sigma_J \leq \max_{w \in V}\ellmax(w) \wedge\sum_{i=1}^J (\sigma_i+\ellmax(v_i)) \geq M_1 \log n}
\label{eq:bound_J_first}
\end{align}

We start by showing that the first term in the inequality above is at most $1/(2n^3)$.
Let $h=\sup \{j\geq 1: \sum_{i=1}^j \ellmax(v_i) \geq 7 \log n\}$ and notice that, since $\min_{v \in V}\ellmax(v) \geq 1$, from the minimality of $h$ we have that $h \leq 7 \log n$. Assume that $h \leq J$, otherwise the inequality follows trivially.
\cref{lem:char_platinum}(a) together with \eqref{eq:char_stopping_times} yields 
\begin{align}
& \Probc{\sum_{i=1}^J \ellmax(v_i) \geq 7 \log n \vee \sigma_J \geq \max_{w \in V}\ellmax(w)} \notag \\ &
\leq \Probc{\bigcap_{i=1}^h \{v_i \not \in \Promi_{t_{i+1}} \vee \sigma_i \geq \max_{w \in V} \ellmax(w)\}} \leq \prod_{i=1}^h \left(1-3^{-\eta'_t(v_i)}\right) \leq 2\prod_{i=1}^h \eta'_t(v_i).
\label{eq:upper_bound_J}
\end{align}
Moreover, we have that
\begin{align*}
    \eta'_t(v_i) &\le \sum_{\substack{w \in N(v_i): \\ \ellmax(w)> \ellmax(v_i)}}2^{-\ellmax(v_i)} \leq \deg(v_i) \cdot 2^{-\ellmax(v_i)}
    \\ &
    \leq \frac{\deg(v_i)}{2^{\ellmax(v_i)/2}} 2^{-\ellmax(v_i)/2} \leq 2^{-\ellmax(v_i)/2},
\end{align*}
where the last inequality follows from the fact that $\ellmax(v_i) \geq 2 \log \deg(v_i)$. Hence, from the latter inequality and \eqref{eq:bound_J_first}  we have that
\begin{equation*}
\Probc{\sum_{i=1}^J \ellmax(v_i) \geq 7 \log n \vee \sigma_J \geq \max_{w \in V} \ellmax(w)} \leq 2\!\prod_{i=1}^h \eta_t'(v_i) \leq 2^{-\sum_{i=1}^h \ellmax(v_i)/2+1} \leq \frac{1}{2n^3},
  \label{prob_upperbound_J}
\end{equation*}
where the last inequality follows from the fact that $\sum_{i=1}^h \ellmax(v_i) \geq 7 \log n$. 

We proceed by showing that the term in \eqref{eq:bound_J_first} is bounded by $1/(2n^3)$.
From \eqref{eq:char_stopping_times} and \cref{lem:char_platinum}(b) we have that, for each $i=1,\dots, J-1$, the random variables $\sigma_i$ are stochastically dominated by $\ellmax(v_i)+Y_i$, where $Y_i$ are independent geometric random variables with parameter $1/2$. We have that, fixing $M_1 = 36 +c$ and since $\max_{w \in V} \ellmax(w) \leq c \log n$,
\begin{align*}
&\Probc{\sum_{i=1}^J \ellmax(v_i) \leq 7 \log n \wedge \sigma_J \leq \max_{w \in V}\ellmax(w) \wedge \sum_{i=1}^J \sigma_i +\ellmax(v_i)\geq \!M_1 \log n}
\\ & \leq \Probc{\sum_{i=1}^{J-1}(Y_i + 2\ellmax(v_i)) +\sigma_J \geq \!M_1 \log n \wedge \sum_{i=1}^J \ellmax(v_i) \leq \!7 \log n \wedge \sigma_J\! \leq \max_{w \in V}\ellmax(w)}
\\ &
\leq     \Probc{\sum_{i=1}^{J-1}Y_i \geq 2J + 8 \log n \wedge J\leq \!7 \log n} 
\\ &= \Probc{\mathrm{Bin}(2J+8 \log n,\tfrac{1}{2}) \leq J-1 \wedge J \leq 7 \log n}
\leq \frac{1}{2n^3},
\end{align*}
where the last inequality follows from \cref{lem:geometric_distribution}, in \cref{sec:app-tools}, and since $\sum_{i=1}^J\ellmax(v_i) \leq 7 \log n$ implies that $J \leq 7 \log n$. 

\subparagraph{Proof of (ii).}

This time we are looking at the randomness of the rounds $[t+m_i+1, t+m_i+\tau_{i+1}]$ for $i=1,\dots, J$.
From \eqref{eq:char_stopping_times} and \cref{lem:lower_bound_platinum}, we have that the random variables $\tau_i$ are stochastically dominated by
$2\gamma^{-1} \ellmax(v_i) + X_i$,
where $X_i$ are i.i.d.~geometric random variable with parameter $p = 1-e^{-\gamma}$, where $\gamma = e^{-30}$. Then, we have that, assuming that $\sum_{i=1}^J (\sigma_i + \ellmax(v_i)) \leq M_1 \log n$ and in particular that $\sum_{i=1}^J \ellmax(v_i) \leq M_1 \log n$, if $M_2 = 2 \gamma^{-1}M_1 + M_1/p + 4/p^2$, then
\begin{align}
&\Probc{\sum_{i=1}^J \tau_i \geq M_2 \log n \mid \sum_{i=1}^J \ellmax(v_i) \leq M_1 \log n}  \notag
\\ &\leq \Probc{\sum_{i=1}^J X_i + 2 \gamma^{-1}\ellmax(v_i) \geq M_2 \log n \mid \sum_{i=1}^J \ellmax(v_i) \leq M_1 \log n}
\notag
\\
    & \leq \Probc{\sum_{i=1}^J X_i \geq \tfrac{J}{p}+\tfrac{4\log n}{p^2} \mid \sum_{i=1}^J \ellmax(v_i)\leq M_1 \log n} \notag
    \\ &= \Probc{\text{Bin}(\tfrac{J}{p}+\tfrac{4\log n}{p^2},p) \leq J \mid \sum_{i=1}^J \ellmax(v_i) \leq M_1 \log n}  \leq \frac{1}{n^3}\label{eq:geometric_concentration}
\end{align}
where \eqref{eq:geometric_concentration} follows from \cref{lem:geometric_distribution}, and the last inequality follows from the fact that $p = 1-e^{-\gamma}$ and that $\sum_{i=1}^J \ellmax(v_i) \leq M_1 \log n$ implies that  $J \leq M_1 \log n$.
\end{proof}

We now can proceed with the proof of \cref{thm:jeavons_degree}.

\begin{proof}[Proof of \cref{thm:jeavons_degree}]
We have $2\log \deg(w) + 30 \leq \ellmax(w) \leq c_2 \log n$ for every $w \in V$ and some $c_2=O(1)$. For each $i =1,\dots, c_2\log \log n$, we consider the following subsets of vertices
\[V_i = \{v \in V: \ellmax(v) \in [2^i, 2^{i+1}]\}.\]
Let $T_i$ be the round until all the vertices in $\cup_{j \leq i}V_j$ are stabilized, i.e.,
\[T_i = \min\{t' \geq 1: \cup_{j \leq i}V_j \subseteq S_{t'}\}.\]
We have that, for each vertex $v \in V_{i+1}$ and each $t \geq T_i$,
\[ 2\ellmax(u) \geq \ellmax(v) \quad \forall u \in N^+(v) \setminus S_{t}.\]
Indeed, we have $u \not \in \cup_{j \leq i}V_j$ since $u \not \in S_t$ and $t \geq T_i$. Hence, $\ellmax(u) \geq 2^{i+1}$. Since $v \in V_{i+1}$, $\ellmax(v) \leq 2^{i+2}$ and so $2\ellmax(u) \geq \ellmax(v)$. We also have, for each $t \geq T_i$ and each $v \in V_{i+1}$
\[\eta_t(v) \leq \sum_{u \in N(v)\setminus S_t} 2^{-\ellmax(u)} \leq \sum_{u \in N(v) \setminus S_t} 2^{-\ellmax(v)/2} \leq \sum_{u \in N(v)\setminus S_t} \frac{1}{\deg(v)} 2^{-15} \leq 0.0001,\]
where the second inequality follows from the fact that $2\ellmax(u) \geq \ellmax(v) $, and the third inequality since $\ellmax(v) \geq 2 \log \deg(v) + 30$. 

We can now apply \cref{lem:stabilization}, if $t \geq \max\{ T_i, \max_{w \in V} \ellmax(w)+1\}$, to all the vertices $v \in V_{i+1}$, obtaining (with an union bound over all the vertices in $V_i$) the existence of a round $m_i = O(\log n)$ such that $\Probc{V_{i+1}\subseteq S_{t+m_i}} \geq 1-1/n$. Applying this argument iteratively for each $i=1,\dots, c_2\log \log n$, we obtain the existence of a round 
$$m= \sum_{i=1}^{\log \log n}m_i = O(\log n \cdot \log \log n)$$ 
such that all  vertices are stable w.h.p.\ at round $m$. 
\end{proof}

\section{Proof of Key Lemmas}
\label{sec:proof_main_lemmas}
\subsection{Lower Bound on Platinum Rounds (Proof of \texorpdfstring{\cref{lem:lower_bound_platinum})}{Lemma Lower Bound Platinum}}
\label{sec:lem:lower_bound_platinum}

Before proving \cref{lem:lower_bound_platinum}, we introduce some  definitions and preliminary lemmas.

\begin{definition}[Light Vertices] A vertex $v \in V$ is called \emph{light in round $t$} if $\mu_t(v) >0 \wedge (d_t(v) \leq 10 \vee \ell_t(v) \leq 0).$
We denote with $L_t$ the set of light vertices at round $t$ and with $H_t = V \setminus L_t$ the set of \emph{heavy (non-light) vertices at round $t$}.
\label{def:light_vertices}
\end{definition}

Intuitively, a light vertex $v$ is prominent or has a positive, constant probability of not receiving a beep signal during round $t$ and, in the latter case, if $p_t(v)$ is large enough, $v$ has a constant probability of beeping without beeping neighbors during round $t$.
We remark that the condition $\mu_t(v)>0$ is necessary since, if $\mu_t(v)=0$, the vertex $v$ hears a beep during round $t$ with probability $1$.

We denote with $d_t^L(v)=\sum_{u \in N(v) \cap L_t}p_t(u)$ the expected number of beeping light neighbors of $v$ in round $t$, and with $d_t^H(v)=\sum_{ u \in N(v) \cap H_t}p_t(u)$ the expected number of beeping heavy neighbors of $v$ in round $t$. We notice that $d_t(v)=d_t^L(v)+d_t^H(v)$.

\begin{definition}[Golden Rounds]
\label{def:golden}
    Round $t$ is a \emph{golden round of vertex $v$} if one of the following two conditions are satisfied:
    \begin{enumerate}[(a)]
        \item $\ell_t(v) \leq 1$  and $d_t(v) \leq 0.02$;
        \item $d_t^L(v) > 0.001$.
    \end{enumerate}
    We denote with $G_{t,k}(v)$ the number of golden rounds of vertex $v$ during rounds $\{t,\dots,t+k\}$.

\end{definition}

In the next section, we will give a lower bound on the number of golden rounds.

\subsubsection{Lower Bound on Golden Rounds}

\begin{lemma}
\label{lem:golden_rounds}
Assume that $\ellmax(w) \geq \log \deg(w) + 4$ for all $w \in V$. Consider a vertex $v \in V$ and a round $t >\max_{w \in V}\ellmax(w)$ such that $t$ is not a platinum round of $v$, and $\eta_t(v) \leq 0.0001$. Let $\tau^{(v)}(t)$ be defined as in \cref{lem:lower_bound_platinum}. Then, we have that, for any $k \geq 70 \cdot \ellmax(v)$,
\[\Probc{G_{t,k}(v) \leq 0.05k \wedge \tau^{(v)}(t) > k \mid \mathcal{F}_t} \leq e^{-k/100}.\]
\end{lemma}

We notice that, if round $t$ is not a platinum round of $v$, every round $s \in [t,\tau^{(v)}(t)]$ is also not a platinum round of $v$, since the only way a vertex in $N^+(v)$ can take a negative level is by beeping without beeping neighbors, and $\tau^{(v)}(t)$ is the first round that happens. 
The proof of \cref{lem:golden_rounds} relies on the following result.

\begin{lemma}
Let $v\in V$ and $t > \max_{w \in V}\ellmax(w)$ such that round $t$
is not a platinum round of $v$ and $\eta_t(v) \leq 0.0001$.
\begin{enumerate}[(a)]
\item If $d_t(v) \leq 0.02$, then $\ell_{t+1}(v)\leq \max\{1,\ell_t(v) - 1\}$ with probability at least $0.97$.
\item If $d_t(v) >0.01$ and $d_t^L(v) <0.01d_t(v)$, then with
probability at least 0.97, we have that $d_{t+1}(v) <0.6d_t(v)$ or that
$t+1$ is a platinum round for $v$.
\end{enumerate}
\label{claim:wrong_moves}
\end{lemma}

\begin{proof}
We fix the execution up to the end of round $t$, so we do not have to condition probabilities on $\mathcal{F}_t$. In what follows, we prove separately the two statements.

We prove (a) first. 
Since $d_t(v)=\sum_{u \in N(v)}p_t(u)\leq 0.02$ it follows that $p_t(u) \leq \frac{1}{2}$ for all $u \in N(v)$. Thus, the probability that no neighbor of $v$ beeps is at least $\prod_{u \in N(v)}\left(1-p_t(u)\right) \geq 4^{-d_t(v)} \geq 0.97$. Hence, $\Prob[\ell_{t+1}(v) \leq \max\{\ell_t(v) - 1,1\}]\ge 0.97$.

Next we prove (b).  Since $t$ is not a platinum round of  $v$, we have that for each $u \in N^+(v)$, $\ell_t(u) \geq 1$. Moreover, we notice that there may be in round $t$ a beeping vertex $u \in N^+(v)$ that does not receive a signal, and so $\ell_{t+1}(u) =-\ellmax(u) \leq 0$.

    For any vertex $u \in N^+(v)$, we have the following upper bounds for $p_{t+1}(u)$ (recall that $\ell_t(u) >0$ since $t$ is not a platinum round of $v$):
    \[
        p_{t+1}(u) \leq
        \begin{cases}
            2^{-\ellmax(u) + 1} &\text{if }\ell_t(u)=\ellmax(u) \text{ and } u \not \in S_t\\
            0 &\text{if }\ell_t(u)=\ellmax(u) \text{ and } u \in S_t \\
            \frac{p_t(u)}{2} &\text{if } B_t(u) \geq 1 \text{ and } \ell_t(u) \neq \ellmax(u) \\
            2p_t(u) &\text{if } B_t(u) = b_t(u) = 0 \text{ and } \ell_t(u) \neq \ellmax(u) \\
            1 &\text{if } B_t(u) = 0, b_t(u) = 1 \text{ and } \ell_t(u) \neq \ellmax(u)
        \end{cases}  
    \]
    The last case, i.e., when $B_t(u)=0$ and $b_t(u)=1$ implies that $t+1$ is a platinum round for $v$, and that $\ell_{t+1}(u)=-\ellmax(u)$.
    Define $J_{t+1}(v)$ the set of such vertices, i.e., the set of vertices in $N(v)$ beeping in round $t$ without beeping neighbors. Then,
    \begin{align*}
        d_{t+1}(v) &\leq \sum_{\substack{u \in N(v)\setminus S_t: \\ \ell_t(u)= \ellmax(u)}}2^{-\ellmax(u)+1} + \sum_{\substack{u \in N(v): \\ B_t(u) \geq 1 \\ \ell_t(u)\not= \ellmax(u)}}\frac{p_t(u)}{2} + \sum_{\substack{u \in N(v): \\ B_t(u)=b_t(u)=0}}2p_t(u) + J_{t+1}(v)
        \\
        & \leq 2\eta_t(v) + \sum_{u \in N(v) \cap H_t}p_t(u)\left(\frac{1}{2}+2\cdot \mathds{1}_{B_t(u)=0}\right) + \sum_{u \in N(v) \cap L_t}2p_t(u) + J_{t+1}(v).
    \end{align*}
We notice that, since $d_t^L(v) = \sum_{u \in N(v)\cap L_t}p_t(u) \leq 0.01d_t(v)$ and $\eta_t(v) \leq 0.001$, we have that
\[d_{t+1}(v) \leq 0.0002 + 0.02d_t(v) + J_{t+
1}(v) + \sum_{u \in N(v) \cap H_t}p_t(u)\left(\frac{1}{2}+2\cdot \mathds{1}_{B_t(u)=0}\right). \]
We want to bound, for each $u \in N(v) \cap H_t$, the probability that $B_t(u)=0$. Since $u \in N(v) \cap H_t$ and $\ell_t(u) \geq 1$, then $d_t(u) \geq 10$ or $\mu_t(u)=0$. In the latter case, we know that $u$ has some neighbor $u' \in N(u)$ with $p_t(u')=1$. Hence, we have that $\Probc{B_t(u)=0}=0$. In the former case, we have that none of $u$'s neighbors is beeping with probability at most
\[\prod_{w \in N(u)}(1-p_t(w)) \leq e^{-d_t(u)} \leq e^{-10}.\]
Hence, we have that, for each $u \in N(v) \cap H_t$, $\Probc{B_t(u)=0} \leq e^{-10}$. So,
\[\Expc{\sum_{u \in N(v) \cap H_t}2p_t(u) \mathds{1}_{B_t(u)=0}} \leq \sum_{u \in N(v) \cap H_t}2p_t(u)e^{-10}. \]
 Markov's inequality implies that $\sum_{u \in N(v) \cap H_t}2p_t(u) \mathds{1}_{B_t(u)=0} \leq 0.01 \sum_{u \in N(v) \cap H_t}2p_t(u)$ with probability at least $1-\frac{e^{-10}}{0.01} \geq 0.97$.
Thus, with probability at least $0.97$, we have that
\begin{equation}
    d_{t+1}(v) \leq 0.0002 + 0.02d_t(v) + J_{t+1}(v) + 0.5d_t(v) + 0.02d_t(v) \leq 0.6d_t(v) + J_{t+1}(v),
    \label{eq:bound_d_t}
\end{equation}
where the last inequality follows by noticing that $d_t(v) >0.01$ and hence $0.0002 < 0.02d_t(v)$. This yields the lemma, since  \eqref{eq:bound_d_t} implies that either $d_{t+1}(v)<0.6d_t(v)$, or $J_{t+1}(v) >0$, and hence $t+1$ is a platinum round for $v$.
\end{proof}

We are now ready to prove \cref{lem:golden_rounds}.

\begin{proof}[Proof of \cref{lem:golden_rounds}]
Fix a vertex $v \in V$.    We consider $k \geq 70\ellmax(v)$ consecutive rounds, starting from a round $t$ which is not a platinum round of $v$. Since $\eta_t(v)$ is decreasing in $t$, in all rounds $t +m$, $m\geq 0$, we have $\eta_{t+m}(v)<0.0001$. We consider the following sets of rounds
\begin{gather*}
D_{t,k}(v)=\{0\leq m\leq k: d_{t+m}(v)>0.2\}    
\\
E_{t,k}(v)=\{0 \leq m \leq k:d_{t+m}(v)>0.1 \text{ and } d_{t+m}^L(v) \geq 0.1d_{t+m}(v)\}
\\
F_{t,k}(v)=\{0 \leq m \leq k:d_{t+m}(v)>0.1 \text{ and }d_{t+m}^L(v)<0.1d_{t+m}(v)\}
\\
H_{t,k}(v)=\{0 \leq m \leq k: d_{t+m}(v) < 0.2 \text{ and } \ell_{t+m}(v) \leq 1\}.
\end{gather*}
We say that in some round $t'$ we have a \emph{wrong move} if none of the following conditions occurs
\begin{enumerate}[(a)]
        \item $t$ or $t+1$ is a platinum round of $v$;
        
        \item $\eta_t(v)>0.0001$;
        
        \item $d_t(v)\leq 0.01$ or $d_t(v) > 0.02$;

        \item $d_t^L(v)\geq 0.01d_t(v)$;

                \item $d_{t+1}(v)<0.6d_t(v)$;
        
        \item  $\ell_{t+1}(v)\leq \max\{
        \ell_t(v)-1,1\}$;
        
    \end{enumerate}

From \cref{claim:wrong_moves}, we have that a vertex makes a wrong move with probability at most 0.03. Since the randomness of each round is independent of the others, we know by Chernoff's bound, that in the rounds $\{t,t+1,\dots, t+k\}$ there are at most $0.04k$ wrong moves with probability at least $1-e^{-k/100}$, and we will refer to this event with $B$.

In the rest of the proof, we assume that $B$ happens, and we will see that it implies, deterministically, that $\tau^{(v)}(t) \leq k$ or that $G_{t,k}(v) \geq 0.1k$. So, we assume that $\tau^{(v)}(t) > k$ and we will prove that, under event $B$, this implies that $G_{t,k}(v) \geq 0.1k$. We remark that, if $\tau^{(v)}(t) >k$, for each $0\leq m\leq k$, we have that $d_{t+m+1}(v) \leq 2d_{t+m}(v)$. 

In what follows, we will prove that:
\begin{enumerate}[(i)]
    \item if $E_{t,k}(v)<0.05k$, then $D_{t,k}(v) < 0.25k$
    \item if $D_{t,k}(v)<0.25k$, then $H_{t,k}(v)>0.28k$.
\end{enumerate}

We prove (i) first.
We denote with $D_{t,k}'(v)$ the set $\{0\leq m \leq k: d_{t+m}(v)>0.1\}$ and we notice that $D'_{t,k}(v) = E_{t,k}(v) \cup F_{t,k}(v)$. Also let $h=|D_{t,k}(v)|$ and $h'=|D'_{t,k}(v)|$. Since the number of wrong moves is bounded by $0.04k$, and since $E_{t,k}(v) <0.05k$ the number of rounds in $D'_{t,k}(v)$ in which $d_{t+m}(v)$ can double is at most $0.09k$, and in the rest of the rounds it will decrease of a factor of $0.6$. 

In order to keep $d_{t+m}(v)>0.2$, in a consecutive interval of rounds in $D_{t,k}'(v)$, the number of increasing moves must be at least $\log_{0.5}(0.6) >0.7$ times the number of decreasing moves, and at most $\log_{5/3}(10d_t(v)) \leq 2 \log(10\deg(v)) \leq 2 \log \deg(v) + 8$ decreases are used to decrease the initial value of $d_t(v)$ below 0.1. Hence, the total number of rounds in $D_{t,k}(v)$ is at most
\begin{align*}
    0.09k + \frac{0.09}{0.7}k + 2\log(\deg(v)) + 8 &\leq 0.22k + 2\log \deg (v) + 8 \\ &\leq 0.22k + 2\ellmax(v) 
    \\& 
    \leq 0.22k + 0.03k
    \\& 
    = 0.25k.
\end{align*}

Next we prove (ii).
Since $|D_{t,k}(v)| \leq 0.25k$, the set $D_{t,k}^C(v)=\{0 \leq m \leq k:d_{t+m}(v)\leq 0.2\}$ contains at least $0.75k$ rounds. The number of wrong moves is bounded by $0.04k$, and in rounds $D_{t,k}^C(v)$ a wrong moves implies that $\ell_{t+m+1}(v) = \min\{\ell_{t+m}(v) + 1, \ellmax(v)\}$. Moreover, we have that in the rounds $D_{t,k}(v)$,
$\ell_{t+m+1}(v) \leq \min\{\ell_{t+m}(v)+1,\ellmax(v)\}$ is satisfied. Hence, $\ell_t(v)$ can increase in at most $|D_{t,k}(v)| + 0.04k \leq 0.29k$ rounds. The rounds in $D_{t,k}^C(v)$ in which no wrong move occurred are such that $\ell_{t+m+1}(v) = \max \{1,\ell_{t+m}(v)-1\}$, since we assumed that $\tau^{(v)}(t)> k$. Since $D_{t,k}^C(v)$ has at least $0.75k$ elements, and since the number of wrong moves is bounded by $0.04k$, the number of moves in which $\ell_{t+m}(v)$ decreases is at least $0.75k-0.04k= 0.71k$. Since the number of rounds in which $\ell_{t+m}(v)$ increases is at most $0.29k$, we have that the number of increases is at least $2.4$ times the number of decreases.

Denote the number of rounds in $D_{t,k}^C(v)$ where $\ell_{t+m}(v)$ decreases by $U$ and those where it increases by $D$. Thus, $D + U \geq 0.75k$ and $U \geq 2.4D$. In the worst case, each round with an increase follows a round with a decrease. Then, we still have $0.75k - 2D$ rounds with an increase left. Then, $0.75k-2D = U - D \geq 0.58U \geq
0.3k$. As it takes at most $\ellmax(v)$ for $p_{t+m}(v)$ to reach 1/2 we can say that, since $k \geq 70\ellmax(v)$,
we have at least $0.3k - \ellmax(v) > 0.28k$ rounds where $\ell_{t+m}(v)=1$ and $d_{t+m}(v) <0.2$, hence $H_{t,k}(v) > 0.28k$.
\end{proof}

\subsubsection{From Golden to Platinum Rounds}

We first notice that, for each golden round $t$ of $v$, round $t+1$ is platinum for $v$ with constant probability. Indeed, we have the following lemma. 

\begin{lemma}
\label{lem:golden_rounds_joining} Let $t > \max_{w\in V}\ellmax(w)$ be a non-platinum round of $v$, and consider $\tau^{(v)}(t)$ as in \cref{lem:lower_bound_platinum}. Then, we have that, for each $t\leq s < \tau^{(v)}(t)$ which is golden for $v$, $\Probc{\tau^{(v)}(t) = s+1 \mid \mathcal{F}_s} \geq \gamma$, where $\gamma \geq e^{-27} $.
\end{lemma}

\begin{proof}
Since $t \leq s < \tau^{(v)}(t)$, $s$ is not a platinum round of $v$, every vertex $u \in N^+(V)$ is such that $\ell_s(v) \geq 1$. In what follows, we prove that, with constant probability $\gamma>0$, during round $s$, there is a vertex $u \in N^+(v)$ such that $B_s(u)=0$ and $b_s(u)=1$. This implies that $\ell_{s+1}(u)=-\ellmax(u)$ and that $\mu_{s+1}(u)>0$, hence that $s+1$ is platinum for $v$.
Since $s$ is golden for $v$, we have that part (a) or (b) of \cref{def:golden} holds.

First, assume that (a) holds, thus $\ell_s(v) \leq 1$ and $d_s(v) \leq 0.02$. In this case, with probability at least $0.48$, we have that $B_s(v)=0$ and $b_s(v)=1$ and so $s+1$ is platinum for $v$. Indeed, the expected number of beeping neighbors of $v$ during round $s$ is $d_s(v) \leq 0.02$. Therefore, for Markov's inequality, $\Probc{B_s(v)\geq 1}\leq 0.02$, and $v$ beeps with probability at least $1/2$, and then the level of $v$ is updated to zero with probability  at least $\frac{1}{2}\cdot 0.98 > 0.48$.

We now assume that round $s$ satisfies (b), therefore that $d_s^L(v) \geq 0.001$. We will prove that, in such types of rounds, with probability at least a constant $\gamma$, there is a beeping vertex $u \in N(v)$ with no beeping neighbors during round $s$. 
Let $k = |N(v) \cap L_s| $ be the number of light neighbors of $v$, and denote $\{w_1,\dots,w_k\}=N(v)\cap L_s$. We remark that all  light vertices $w_i$ are such that $\ell_s(u)>0$ for each $u \in N(w_i)$ and $d_s(w_i) \leq \maxd$.  
We define $\mathcal{E}_i$ to be the event indicating that  vertex $w_i$ is beeping during round $s$. Let $\mathcal{E}=\cup_i \mathcal{E}_i$.
We have that,
\[\Pr\left[\mathcal{E}\right] \geq 1- \prod_{\substack{w \in N(v) \cap L_s}}\left(1-p_t(w)\right)\geq 1-e^{-\sum_{\substack{w \in N(v) \cap L_s}}p_s(w)} = 1-e^{-d_s^L(w)} \geq 1- e^{-0.001} .\]

Suppose that $\mathcal{E}$ occurs, and let $j$ be the smallest index such that $\mathcal{E}_j$ occurs, i.e., $\Bar{\mathcal{E}}_1\cap \Bar{\mathcal{E}}_2 \cap \cdots \cap \Bar{\mathcal{E}}_{j-1} \cap \mathcal{E}_j$ occurs. If $G_j = N(w_j) \setminus \{w_1,\dots,w_{j-1}\}$, then the probability that no neighbor of $w_j$ in $G_j$ beeps is at least
\[\prod_{u \in G_j}\left(1-p_s(u)\right)\geq \prod_{u \in N(w_j)}(1-p_s(u)) \geq e^{-2d_s(w)} \geq e^{-20}.\]
where the first inequality follows from the fact that, since $w_j$ is light, $\mu_s(w_j)>0$ and so each $u \in N(w_j)$ is such that $\ell_s(u) \geq 1$.
Combining this with the previous inequality, we have that a vertex $w \in N(v)$ with $d_s(w) \leq \maxd$ is beeping with no beeping neighbors with probability at least 
$e^{-20}(1-e^{-0.001}) > e^{-27}=\gamma$.
\end{proof}

\cref{lem:lower_bound_platinum} follows from \cref{lem:golden_rounds,lem:golden_rounds_joining}.

\begin{proof}[Proof of \cref{lem:lower_bound_platinum}]
We fix the execution up to the end of round $t$, so we do not have to condition on $\mathcal{F}_t$. We have that
\begin{align}
    &\Probc{\tau^{(v)}(t) >k} \notag \\
    &= \Probc{\tau^{(v)}(t) >k \wedge G_{t,k}(v) > 0.05k} + \Probc{\tau^{(v)}(t) > k \wedge G_{t,k}(v) \leq 0.05k}
    \notag
    \\ & \leq \Probc{\tau^{(v)}(t)>k \wedge G_{t,k}(v) >0.05k} + e^{-k/100}
    \label{eq:using_lb_golden}
    \\ & \leq (1-e^{-27})^{0.05k} + e^{-k/100}\label{eq:using_golden_and_platinum} \\ &
    \leq e^{-e^{-27}0.05k} + e^{-k/100}
     \leq e^{-e^{-29}k} + e^{-k/100}\leq e^{-e^{-30}k}, \notag
\end{align}
where \eqref{eq:using_lb_golden} follows from \cref{lem:golden_rounds}, and \eqref{eq:using_golden_and_platinum} follows from \cref{lem:golden_rounds_joining}.
\end{proof}

\subsection{Stopping Times for Platinum Rounds (Proof of \texorpdfstring{\cref{lem:char_platinum})}{Lemma Stopping Time Platinum}}
\label{sec:platinum}

\begin{proof}[Proof of \cref{lem:char_platinum}] 
We fix the execution up to the end of round $t$, so we do not have to condition probabilities on $\mathcal{F}_t$. 

We prove part (a) first.
Since $u \in \Promi_t\setminus S_t$, we have $\ell_t(u)\leq 0$. Since $t >\max_{w \in V} \ellmax(w)$ \cref{fact:no_adjacent_negative} implies that $\mu_t(u)>0$, i.e., $\ell_t(w)>0$ for all $w\in N(v)$. By \cref{fact:no_beeping_neighbors} there exist a round $t'\leq t$ with $\ell_{t'}(u)= -\ellmax(u)$ and $t-\ellmax(u) \le t'$. Thus, each neighbor $w$ of $u$ incremented its level during the rounds $t', t'+1, \ldots, t$ or the level of $w$ reached $\ellmax(u)$.  Let $\ell=t-t'$. Thus, $\ell_t(w) \ge \min\{\ellmax(w), \ellmax(u)-\ell\}$. Hence, if $\ell_t(w)< \ellmax(w)$ then $p_t(w) \leq 2^{-(\ellmax(u)-\ell)}$. This yields
\[d_t(u) = \sum_{w \in N(u)\setminus S_t} p_t(w) \leq \sum_{w \in N(u)\setminus S_t}2^{-(\ellmax(u)-\ell)},\]
and also that, in the subsequent $\ell$ rounds, vertex $u$ is beeping and the level of each of $u$'s neighbors increases in each round. Therefore, we have $\ell_{t+\ell}(w) \geq \min\{\ellmax(w),\ellmax(u)\}$  for each $w \in N(u)$, and, moreover
\[d_{t+\ell}(u)= \sum_{\substack{w \in N(u)\setminus S_t:\\ \ellmax(w)> \ellmax(u)}}2^{-\ellmax(u)} = \eta'_t(u).\]

We notice that, if $\ellmax(u) \geq \ellmax(w)$ for each $w \in N(u)$, then we have that $d_{t+\ell}(u)=0$ and hence $\Probc{\sigma^{(u)}(t) = \sigmastable^{(u)}(t) \wedge \sigma^{(u)}(t) \leq \ellmax(u)}=1$, and this proves (a) when $\eta_t'(u)=0$.
If otherwise $\eta_t'(u)>0$, we can define
\[\ell' = \max_{w \in N(u)}\ellmax(w)-\ellmax(u)\]
which is such that $0<\ell' \leq \max_{w \in N(u)}\ellmax(w)$, and we have that $(\sigma^{(u)}(t) =\sigmastable^{(u)}(t)) \wedge (\sigma^{(u)}(t) \leq \ell')$ happens with probability at least
\begin{align*}
 \prod_{i=1}^{\ell'}\prod_{\substack{w \in N(u): \\ \ellmax(w)>\ellmax(u)}}\left(1-\frac{p_{t}(w)}{2^{i+\ellmax(u)}}\right) &\geq \prod_{i=1}^{\ell'} \prod_{\substack{w \in N(u): \\ \ellmax(w)>\ellmax(u)}} 3^{-p_{t}(w)2^{-(i+\ellmax(u))}}
 \\ &\geq \prod_{i=1}^{\ell'} 3^{-\eta_t'(u)2^{-i}}
 \\& 
 \geq 3^{-\eta'_t(u)},
\end{align*}
where the first inequality follows from the fact that, for each $w \in N(u)$, $p_t(w)/2^{\ellmax(u)} \leq 2^{-4}$.

Next we prove part (b).
We observe that, for each $x \geq 0$, we have that
\begin{align*}
    &\Probc{\sigma^{(u)}(t) \neq \sigmastable^{(u)}(t) \wedge \sigma^{(u)}(t) > \ellmax(u)+x} 
    \\ &
    \leq \Probc{\sigma^{(u)}(t)=\sigmafree^{(u)}(t) \mid \sigma^{(u)}(t)>\ellmax(u) + x}.
\end{align*}
Since we have that the event $\sigma^{(u)}(t)>\ellmax(u)+x$ implies that, for each $w \in N(u)$, $p_{t+\ellmax(u)+x}(w) \leq 2^{-(\ellmax(u)+x)}$, we have that
\begin{align*}
    \Probc{\sigma^{(u)}(t)= \sigmafree^{(u)}(t) \mid \sigma^{(u)}(t) > \ellmax(u)+x} &\leq 1-\prod_{w \in N(u)\setminus S_t}\left(1-2^{-(\ellmax(u)+x)}\right)
    \\ & \leq \sum_{w  \in N(u) \setminus S_t} 2^{-(\ellmax(u)+x)} 
    \\ & 
    \leq \eta'_t(u) 2^{-x}.
    \qedhere
\end{align*}
\end{proof}

\section{Two Beeping Channels (Proof of \texorpdfstring{\cref{thm:two_beeping_channels}}{Corollary 4})}
\label{sec:two_beeping_channels}

\begin{algorithm}[t]
    \state{$\ell \in\{0,\ldots, \ellmax(v)\}$}
    
    \BlankLine
    
    \InEachRound{$t=1,2,\ldots$}{

        \uIf{$0<\ell < \ellmax(v)$}
            {
            $beep_1 \,\gets\, \true$ with probability 
            $2^{-\ell}$            and $beep_1 \,\gets\, \false$ otherwise\;} 
        \lElse{$beep_1 \,\gets\, \false$}
        {$beep_2 \, \gets \, (\ell = 0)$\;}

        \BlankLine

        \lIf{$beep_1$ or $beep_2$}{send the corresponding signal to all neighbors}
        receive any signals sent by neighbors\;
        \BlankLine

        \BlankLine
        
        \uIf {$beep_2$ signal received}
        {$\ell \, \gets \, \ellmax(v)$;}
        \uElseIf {$beep_1$ signal received}
            {
            $\ell \,\gets\, \min\{\ell+1,\,\ellmax(v)\}$\;}
        \uElseIf {$beep_1$}
            {
            $\ell \,\gets\, 0$\;}
        \ElseIf{$beep_2 = \false$}
            {%
            $\ell \,\gets\, \max\{\ell-1, 1\}$}
    }
    
    \caption{Self-stabilizing version of Jeavons at al.'s algorithm with two beeping channels}
    \label{alg:jsx-ss-twochannels}
\end{algorithm}

One of the reasons that the MIS algorithm of Jeavons et al.\ \cite{JeavonsS016} is not self-stabilizing is the usage of phases consisting of two rounds. This allows a newly joined MIS vertex to signal this event to all neighbors in the subsequent round. Afterwards, these vertices become passive, i.e., no longer participate in the algorithm. Thus, a vertex $v$ that newly joined the MIS cannot be forced by a neighbor that is unaware that $v$ joined the MIS to leave the MIS again in the successive round. This problem can be circumvented if a second beeping channel is available, since we can let the vertices joining the MIS beep on that channel. Indeed, beginning in the round immediately following the round vertex $v$ joined the MIS, it signals in every round on this second channel. Neighbors of $v$ take this as an opportunity to become non-MIS vertices.  This second channel and the corresponding behavior can be easily integrated into \cref{alg:jsx-ss} (see \cref{alg:jsx-ss-twochannels}). The range of state variable $\ell(v)$ is restricted to $[0,\ellmax(v)]$, where $\ell(v)=0$ (resp.\ $\ell(v)=\ellmax(v)$) implies that $v$ is in the MIS (resp.\ not in the MIS). If a vertex $v$ which is enabled to signal with $beep_1$ receives neither signal from a neighbor then it sets $\ell(v)$ to $0$ and signals $beep_2$ in all future rounds. Vertices receiving a $beep_2$ signal set  $\ell(v)$ to $\ellmax(v)$ and refrain from beeping in future rounds. We end this section by giving the proof of \cref{thm:two_beeping_channels}.

\begin{proof}[Proof of \cref{thm:two_beeping_channels}]
We consider \cref{alg:jsx-ss-twochannels} and we notice that the update rule of $\ell$ of the non-stable vertices is the same of \cref{alg:jsx-ss}, and hence we can still use \cref{lem:lower_bound_platinum}, since it relies just on the update rule for $\ell$. Note the difference between the two algorithms: In \cref{alg:jsx-ss} if the level of a vertex is $0$ or lower then it is guaranteed that it sends a beep. In  \cref{alg:jsx-ss-twochannels} a vertex sends
a $beep_2$ signal if and only if its level is $0$. 

We will prove that the termination time of \cref{alg:jsx-ss-twochannels} is $O(\log n)$, if we take $\ellmax(v)\geq 2\log \deg_2(v) + 15$ for every $v \in V$.
We first notice that, in this case, we have that
\[\eta_1(v) \leq \sum_{u \in N(v)} 2^{-2\log \deg_2(u) - 15} \leq \sum_{u \in N(v)}\frac{1}{\deg^2(v)}2^{-15} \leq 0.0001,\]
and hence, for each $t \geq 1$ and $v \in V$ we have that $\eta_t(v) \leq 0.0001$.

We notice that, for a vertex $v \in V$ to stabilize in \cref{alg:jsx-ss-twochannels}, it suffice to have a platinum round for $v$. Hence, from \cref{lem:lower_bound_platinum} we have that each vertex stabilizes in time $O(\log n)$ with probability at least $1-1/n^2$. The theorem follows from the union bound applied to all vertices.
\end{proof}

\section{Conclusion}
\label{sec:conclusion}


In this paper, 
we describe a new randomized self-stabilizing MIS algorithm using the beeping model requiring each vertex to have only limited knowledge about the topology that comes close to the $O(\log n)$ time bound. The algorithm is motivated by the non self-stabilizing algorithm of Jeavons et al.\ \cite{JeavonsS016}. To transform it into a self-stabilizing algorithm we had to overcome two issues: Firstly, the known initial configuration and secondly, the phase concept. We prove that the global knowledge of the maximum degree is sufficient to obtain a $O(\log n)$ self-stabilizing algorithm. If we rely on the local knowledge of the vertex degree,  the algorithm stabilizes in time $O(\log n \cdot \log \log n)$. It is an open question if this upper bound is tight, or whether the analysis can be improved to obtain the upper bound $O(\log n)$.

We remark that, for a beeping model with two channels, we can easily implement the phases with two rounds with the presence of two beeping channels, and we prove that, in such a case, a self-stabilizing MIS algorithm terminating in $O(\log n)$ relies on the (almost) local knowledge of the 2-hop neighbors.
It is natural to ask whether the local knowledge can be completely removed, obtaining an algorithm for the beeping model (with one or two channels) that computes an MIS in a self-stabilizing way.



\bibliography{disc}

\appendix

\section*{APPENDIX}

\section{Tools}
\label{sec:app-tools}

\begin{lemma}
\label{lem:geometric_distribution}
Let $X_1, \dots, X_n$ be a sequence of i.i.d.\ geometric random variables with success probability $p$. Then, we have that
\[\Probc{\sum_{i=1}^n X_i \geq k} = \Probc{\text{Bin}(k,p) \leq n}.\]
\end{lemma}

\begin{proof}
Asking that $\sum_{i=1}^n X_i \geq k$ is like asking that, in $k$ Bernoulli trials, we have less than $n$ successes.
\end{proof}

\begin{theorem}[Chernoff's Inequality]
\label{thm:additive_chernoff}
Let $X=\sum_{i=1}^n X_i$, where $X_i$ with $i \in [n]$ are independently
distributed in $[0,1]$. Let $\mu=\Expc{X}$ and $\mu_- \leq \mu \leq \mu_+$.
Then:
\begin{enumerate}[(a)]
    \item for  every  $t>0$, it holds
    \[
    \Probc{X>\mu_+ +t}\leq  e^{-2t^2/n}  \quad \text{and} 
 \quad \Probc{X<\mu_- -t}\leq e^{-2t^2/n} ;
    \]
    \item for every $0<\epsilon<1$, it holds
    \[
    \Probc{X>(1+\epsilon)\mu_+}\leq e^{-\frac{\epsilon^2}{3}\mu_+} 
    \ \text{ and } \
    \Probc{X<(1-\epsilon)\mu_-}\leq e^{-\frac{\epsilon^2}{2}\mu_-} . 
    \]
\end{enumerate}
\end{theorem}

\end{document}